\documentclass{sig-alternate}

\begin{document}
%

\title{Relationship between circuit complexity and symmetry
}

\numberofauthors{1}
\author{
\alignauthor
Satoshi Tazawa\titlenote{Email adress is tazawa314[at mark]gmail.com}
       \affaddr{}\\
}

\maketitle
\begin{abstract}
It is already shown that a Boolean function for a NP-complete problem can be computed by a polynomial-sized circuit 
if its variables have enough number of automorphisms. Looking at this previous study from the different perspective gives us 
the idea that the small number of automorphisms might be a barrier for a polynomial time solution for NP-complete problems.

Here I show that by interpreting a Boolean circuit as a graph, the small number of graph automorphisms and 
the large number of subgraph automorphisms in the circuit establishes the exponential circuit lower bound for NP-complete problems.
As this strategy violates the largeness condition in Natural proof, this result shows that $P \neq NP$ 
without any contradictions to the existence of pseudorandom functions.
\end{abstract}

\category{F.2}{Theory of Computation}{Analysis of algorithms and problem complexity}

\terms{Theory}
\keywords{Circuit complexity, graph automorphism, $P \neq NP$}

\section{Preliminary and outline of this paper}
In this paper, unbounded depth Boolean circuits with statdard gates AND($\wedge$), OR($\vee$), NOT($\neg$) are discussed. 
AND gates and OR gates have 2 fan-in and unbounded fan-out. In section 3, a NP-complete Boolean function $f_{k}$ for the $k$-clique problem 
with $n$ vertices is discussed. $f_{k}$ can be written as two different ways,
\begin{equation} f_{k}(e_{12},e_{13},\ldots,e_{(n-1)n})=f_{k}(x_{1},x_{2},x_{3},\ldots,x_{\binom{n}{2}}) \end{equation}
(The former emphasizes the variables in $f_{k}$ as edges, the latter emphasizes the number of variables in $f_{k}$, and both expresssions are used.)\\
Section 3.1 shows that the proof strategy in Section3 is non-Naturalizable.\\
In section 4, the relationship to other open problems in computational complexity theory is discussed.

\section{Introduction}
Many approaches have been proposed to solve the famous $P \neq NP$ problem\cite{cook:npcomplete}\cite{karp}\cite{levin}.
Among them, circuit complexity has been studied in order to separate the complexity classes. 
As the exponential circuit lower bound for some NP-complete problem means $P \neq NP$, much effort is devoted to show such lower bound.
Although no exponential circuit lower bound for NP-complete problems is known for general circuit, the exponential circuit lower bound 
for problems in NP is obtained with restrictions on its depth, kinds of gates available, and so on
\cite{Raz1}\cite{Raz2}\cite{Raz3}\cite{Andreev1}\cite{Andreev3}\cite{Hajnal}\cite{Andreev2}\cite{Yao}\cite{Barrington}\cite{Aspnes}\cite{Hastad2}
\cite{Hastad}\cite{depthAjtai}\cite{depth2FSS}\cite{Karchmer}\cite{Paterson}\cite{Alon}\cite{Nisan}\cite{RRaz}\cite{Smolensky}\cite{Tardos}.
However such attempts cannot be extended to general circuits, and the reason why these attempts fail in general ciruits is discussed 
in Natural proof\cite{Natural}. Natural proof showed that proof strategies which are natural or naturalizable can not succeed in establishing 
the exponential lower bound for NP-complete problems under the assumption that there exist the pseudorandom functions.
As it is widely believed that pseudorandom functions exist, a promising approach needs to be nonnaturalizable(In other words, it needs to 
violate one of the conditions(constructivity, largeness, usefulness) in Natural proof).

Apart from this, it is already shown that a Boolean function for a NP-complete problem 
can be computed by a polynomial-sized circuit if its variables have enough number of automorphisms\cite{Evangelos} and 
many difficult SAT instances do not have symmetries\cite{cook:sat}. 
Looking at this previous study from the different perspective gives us the idea that the small number of automorphisms 
might be a barrier for a polynomial time solution for NP-complete problems.

Here I show that by interpreting a Boolean circuit as a graph, the small number of graph automorphisms (global symmetry) and 
the large number of subgraph automorphisms (local symmetry) in the circuit establishes the exponential circuit lower bound 
for NP-complete problems. As this strategy violates the largeness condition in Natural proof, this result shows that $P \neq NP$ 
without any contradictions to the existence of pseudorandom functions.

\section{Proof}
Before going into the discussion of the Boolean circuit of a NP-complete problem, it is necessary to explain the detailed outline of the proof.
In order to show the exponential circuit lower bound, it is necessary to derive an idea from the following well-known fact.\\
``A $C^{\infty}$ function $f(x)$ can be written as an infinite series
\begin{equation}\label{fx} f(x)=\sum_{n=0}^{\infty}a_{n}x^{n} \end{equation}
If constraints on $f(x)$ are given, for example $f^{(n)}(0)=1$, we can specify the form of (\ref{fx}) 
as $f^{(n)}(0)=n!a_{n}=1 \Leftrightarrow a_{n}=\frac{1}{n!}$. As a result, (\ref{fx}) can be written as
\begin{equation} f(x)=\sum_{n=0}^{\infty} \frac{1}{n!}x^{n} \end{equation}
This function $f(x)$ has a simpler form $f(x)=e^{x}$.''\\
Similary, any Boolean function $f(x_{1},x_{2},\ldots,x_{n})$ can be written in the disjunctive normal form.
\begin{equation}\label{Boolean} f(x_{1},x_{2},\ldots,x_{n})=\bigvee_{j=1}^{m} C_{j} \end{equation}
If constraints on $f(x_{1},x_{2},\ldots,x_{n})$ are given, we can specify the form of $f(x_{1},x_{2},\ldots,x_{n})$. 
$f(x_{1},x_{2},\ldots,x_{n})$ may have a simpler form than expressed in the disjunctive normal form. 
In order to establish the exponential circuit lower bound for a Boolean function of a NP-complete problem, 
it is reasonable to specify the form of $f(x_{1},x_{2},\ldots,x_{n})$ before its size is meausured.
Of course, the size of the Boolean circuit should not be measured in the disjunctive normal form 
as conversion into the disjunctive normal form sometimes results in an exponential explosion in the formula. 
So it is necessary to determine the lower bound of the size of the circuits which are logically equivalent to 
$f(x_{1},x_{2},\ldots,x_{n})$ in the specified form (\ref{Boolean}).

In order to separate the $P/poly$ and $NP$, a Boolean function $f_{k}$ for the $k$-clique problem with $n$ vertices (NP-complete problem) 
is discussed\cite{karp}. A graph with $n$ vertices can be encoded in binary using $\binom{n}{2}$ bits (Each bit represents one of possible edges).
In order to specify the form of $f_{k}$, the symmetry of variables in the Boolean circuit needs to be examined by interpreting the Boolean circuit as a graph.
Formally an automorphism of $f_{k}=f_{k}(x_{1},x_{2},\ldots,x_{\binom{n}{2}})$ is defined as follows.
\begin{multline} For\ a\ permutaion\ \sigma\in S_{\binom{n}{2}},\sigma\in Aut(f_{k})\\
if\ f_{k}(x_{\sigma(1)},x_{\sigma(2)},\ldots,x_{\sigma{\binom{n}{2}}})=f_{k}(x_{1},x_{2},\ldots,x_{\binom{n}{2}}) \end{multline}
A Boolean function $f_{k}$ has no automorphism except trivial automorphisms caused by permutations 
of labels on vertices. That is, for any permutation $\sigma\in S_{n}$, it follows that 
\begin{equation} f_{k}(e_{12},e_{13},\ldots,e_{(n-1)n})=f_{k}(e_{\sigma(1)\sigma(2)},e_{\sigma(1)\sigma(3)},\ldots,e_{\sigma(n-1)\sigma(n)})\end{equation}
$f_{k}$ is relatively asymmetrical based on the fact 
\begin{equation}\frac{The\ number\ of\ autmorphisms\ of\ f_{k}}{The\ number\ of\ all\ possible\ permutations}=\frac{n!}{\binom{n}{2}!}\end{equation}
However this information is not enough to specify the form of $f_{k}$ as there exist many kinds of asymmetrical circuits.
So it is necessary to examine not only the global symmetry but also the local symmetry of the circuit as a graph. Unlike ordinary graphs, 
exchangeability of gates should be taken into considerations. 
Gates AND($\wedge$) and OR($\vee$) should be regarded as exchangeable gates when used in the forms
\begin{equation} a\wedge b(\Leftrightarrow b\wedge a),a\vee b(\Leftrightarrow b\vee a),
\neg a\wedge\neg b(\Leftrightarrow\neg b\wedge\neg a),\neg a\vee \neg b(\Leftrightarrow \neg b\vee \neg a).\end{equation}
A NOT gate($\neg$) works as an inexchangeable gate when used in the forms
\begin{equation} \neg a\wedge b (\not\Leftrightarrow \neg b\wedge a),\neg a \vee b(\not\Leftrightarrow \neg b\vee a) \end{equation} 
In this paper, in order to measure the local symmetry of variables $X$ in a Boolean circuit $f$, 
we express the Boolean function $f$ in the disjunctive normal form and examine the automorphisms of $X$ after applying false values {0} to the remaining variables.

For example, let $f$ denote a Boolean function
\begin{equation}f = (x_{1}\wedge x_{2}\wedge x_{3})\vee(x_{2}\wedge x_{4}\wedge\neg x_{5}) \end{equation}
By restricting variables,\\
(A) For $X=\{x_{1},x_{2},x_{3}\}$,

$f^{restricted}(X)=f^{restricted}(x_{1},x_{2},x_{3})$

$=f^{restricted}(x_{1},x_{2},x_{3},0,0)$

$=(x_{1}\wedge x_{2}\wedge x_{3})\vee(x_{2}\wedge 0 \wedge1)=(x_{1}\wedge x_{2}\wedge x_{3})$.

$Aut(f^{restricted}(x_{1},x_{2},x_{3}))\cong S_{3}$.\\
(B) For $X=\{x_{1},x_{2}\}$,

$f^{restricted}(X)=f^{restricted}(x_{1},x_{2})$

$=f^{restricted}(x_{1},x_{2},0,0,0)$

$=(x_{1}\wedge x_{2}\wedge 0)\vee(x_{2}\wedge 0\wedge 1)=0$.

$Aut(f^{restricted}(x_{1},x_{2}))\cong S_{2}$\\
(C) For $X=\{x_{2},x_{4},x_{5}\}$,

$f^{restricted}(X)=f^{restricted}(x_{2},x_{4},x_{5})$

$=f^{restricted}(0,x_{2},0,x_{4},x_{5})$

$=(0 \wedge x_{2}\wedge 0)\vee(x_{2}\wedge x_{4}\wedge \neg x_{5})=(x_{2}\wedge x_{4}\wedge \neg x_{5})$.

$Aut(f^{restricted}(x_{2},x_{4},x_{5}))\cong S_{2}(x_{2}\ and\ x_{4}\ are\ exchageable).$\\
(D) For $X=\{x_{1},x_{2},x_{3},x_{4}\}$,

$f^{restricted}(X)=f^{restricted}(x_{1},x_{2},x_{3},x_{4})$

$=f^{restricted}(x_{1},x_{2},x_{3},x_{4},0)$

$=(x_{1} \wedge x_{2}\wedge x_{3})\vee(x_{2}\wedge x_{4}\wedge 1)$

$=(x_{1} \wedge x_{2}\wedge x_{3})\vee(x_{2}\wedge x_{4})$

$Aut(f^{restricted}(x_{1},x_{2},x_{3},x_{4}))\cong S_{2}(x_{1}\ and\ x_{3}\ are\ exchageable)$\\

In order to measure the local symmetry of $f_{k}$, let $X_{k}=\{x_{1},x_{2},\ldots,x_{\binom{k}{2}}\}=\{ edges\ among\ a\ k-clique\}$.
As a graph with $n$ vertices has $\binom{n}{k}$ candidate $k$-cliques, 
$X_{k}$ is used as a representitive of $\binom{n}{k}$ candidate $k$-cliques $X_{k}^{1},X_{k}^{2},\ldots,X_{k}^{\binom{n}{k}}$.
For a Boolean function $f_{k}$ for the $k$-clique problem with $n$ vertices, the local symmetry of $f_{k}$ can be expressed as follows.
\newtheorem{theorem}{Theorem}
\begin{theorem} $Aut(f_{k}^{restricted}(X_{k}))\cong S_{\binom{k}{2}}$.\end{theorem}
A proof of theorem1 is shown later. In order to prove theorem 1, it is necessary to understand the relationship between the symmetry
(or asymmetry) of variables and the structure of Boolean function. To reduce the the number of possibilities of structures of 
Boolean functions, the following theorem is useful.
\begin{theorem} For $f_{k}^{restricted}(X_{k})=C_{1}\vee C_{2}\vee\ldots\vee C_{m}$, each one of the clauses, $C_{i}(1\le i\le m)$, has to contain 
all of the variables in $X_{k}$.\end{theorem}
\begin{proof} A method of proof by contradition is used. If $C_{i}(1\le i\le m)$ contains only $l(<\binom{k}{2})$ variables, then two cases are conceivable.\\
(1) $C_{i}$ is satisfiable if the truth values of $l$ variables are appropriately chosen.\\
(2) $C_{i}$ is not satisfiable for any of the truth values.

In case (1), $C_{i}=1$ for $l$ variables with appropriately chosen truth values. So
\begin{multline}\label{alwaystrue} f_{k}^{restricted}(X_{k})=C_{1}\vee C_{2}\vee\ldots\vee C_{i}\vee\ldots\vee C_{m}=1 \end{multline}
But if the variable not used in $C_{i}$ takes $0$, $f_{k}^{restricted}(X_{k})$ should return $0$ as $X_{k}$ does not form a $k$-clique.
This contradicts with (\ref{alwaystrue}).

In case (2), as $C_{i}=0$\\
$f_{k}^{restricted}(X_{k})=C_{1}\vee C_{2}\vee\ldots\vee C_{i-1}\vee C_{i}\vee C_{i+1}\vee\ldots\vee C_{m}$\\
$=C_{1}\vee C_{2}\vee\ldots\vee C_{i-1}\vee 0\vee C_{i+1}\vee\ldots\vee C_{m}$\\
$=C_{1}\vee C_{2}\vee\ldots\vee C_{i-1}\vee C_{i+1}\vee\ldots\vee C_{m}$\\ 
$C_{i}$ does not influence the return value and should be erased. Therefore each one of $C_{i}$ has to contain all of the variables in $X_{k}$.
\end{proof}
By theorem 2, we just need to consider clauses, each one of which contains all of the variables in $X_{k}$. 
Regarding the symmetry of variables in a clause, the following theorem follows.
\begin{theorem} For a clause $C(X_{k})$ in which all of the variables are connected by $\wedge$, it follows that\\
$Aut(C(X_{k}))\cong S_{\binom{k}{2}}$\\
$\Leftrightarrow C(X_{k})=(x_{1}\wedge x_2\wedge\ldots\wedge x_{\binom{k}{2}})\ or\ (\neg x_{1}\wedge \neg x_2\wedge\ldots\wedge \neg x_{\binom{k}{2}})$
\end{theorem}
\begin{proof} To show $\Leftarrow$ is trivial. So it is necessary to show $\Rightarrow$. If $Aut(C(X_{k}))\cong S_{\binom{k}{2}}$, then 
all of the variables in $X_{k}$ are exchageable. Based on the simply observation, $x_{i}$ and $x_{j}$ $(x_{i},x_{j}\in X_{k}, i\neq j)$ are exchageable 
in $C(X_{k})$ if and only if $x_{i}$ and $x_{j}$ take the forms $(x_{i}\wedge x_{j})$ or $(\neg x_{i}\wedge \neg x_{j})$. 
Therefore in order for all of the variables in $C(X_{k})$ to be exchangeable,\\
$C(X_{k})=(x_{1}\wedge x_{2}\wedge\ldots\wedge x_{\binom{k}{2}})\ or\ (\neg x_{1}\wedge \neg x_2\wedge\ldots\wedge \neg x_{\binom{k}{2}})$
\end{proof}
Regarding the asymmtry, many possibilities can be considered. So I discuss the case where one transposition does not follow.
\begin{theorem} For a clause $C(X_{k})$ in which all of the variables are connected by $\wedge$,\\ 
one transposition, say $(x_{1}\ x_{2})(x_{1},x_{2}\in X_{k})$, does not follow\\
$\Leftrightarrow C(X_{k})=(x_{1}\wedge \neg x_{2}\wedge\ remaining\ variables)\ or\ (\neg x_{1}\wedge x_{2}\wedge\ remaining\ variables)$
\end{theorem}
\begin{proof} To show $\Leftarrow$ is trivial. So it is necessary to show $\Rightarrow$.
Based on the simply observation, $x_{1}$ and $x_{2}$ are inexchageable if and only if one of them is connected to a NOT gate. 
As all of the variables in a clause are connected by $\wedge$, 
$C(X_{k})=(x_{1}\wedge \neg x_{2}\wedge\ remaining\ variables)\ or\ (\neg x_{1} \wedge x_{2}\wedge\ remaining\ variables)$\\
\end{proof}
Using these results, theorem 1 is shown here.\\

Theorem 1. $Aut(f_{k}^{restricted}(X_{k}))\cong S_{\binom{k}{2}}$.\\
\begin{proof}
A method of proof by contradiction is used.\\
For $f_{k}^{restricted}(X_{k}))=C_{1}(X_{k})\vee C_{2}(X_{k})\vee\ldots\vee C_{m}(X_{k})$, suppose if $Aut(C_{i}(X_{k}))\not\cong S_{\binom{k}{2}}$, 
then one of the transpositions, say $(x_{1}\ x_{2})(x_{1},x_{2}\in X_{k})$, does not follow in $C_{i}(X_{k})$. 
To satisfy this inexchageability of $x_{1}$ and $x_{2}$, 
\begin{equation}\label{inex1} C_{i}(X_{k})=(\neg x_{1}\wedge x_{2}\bigwedge_{3\le z\le\binom{k}{2}}(\neg x_{z})\bigwedge_{3\le w\le\binom{k}{2}}x_{w})(z\neq w)\end{equation}
\begin{equation}\label{inex2} or\ C_{i}(X_{k})=(x_{1}\wedge\neg x_{2}\bigwedge_{3\le z\le\binom{k}{2}}(\neg x_{z})\bigwedge_{3\le w\le\binom{k}{2}}x_{w})(z\neq w)\end{equation}
However by assigning values $x_{1}=0,x_{2}=1,x_{z}=0,x_{w}=1$ to (\ref{inex1}), (\ref{inex1}) returns $1$ though $f_{k}^{restricted}(X_{k})$ should return 0. 
By assigning values $x_{1}=1,x_{2}=0,x_{z}=0,x_{w}=1$ to (\ref{inex2}), (\ref{inex2}) returns 1 though $f_{k}^{restricted}(X_{k})$ should return 0. 
Therefore $Aut(C_{i}(X_{k}))\cong S_{\binom{k}{2}}$ and $Aut(f_{k}^{restricted}(X_{k}))\cong Aut(C_{1}(X_{k})\vee C_{2}(X_{k})\vee\ldots\vee C_{m}(X_{k}))\cong S_{\binom{k}{2}}$.
\end{proof}

\begin{theorem} $f_{k}^{restricted}(X_{k})=(x_{1}\wedge x_{2}\wedge\ldots\wedge x_{\binom{k}{2}})$ \end{theorem}
\begin{proof}
By theorem1, $f_{k}^{restricted}(X_{k})=C_{1}(X_{k})\vee C_{2}(X_{k})\vee\ldots\vee C_{m}(X_{k})$ and $Aut(C_{i}(X_{k}))\cong S_{\binom{k}{2}}(i=1,2,\ldots,m)$. 
By theorem3, a clause $C_{i}(X_{k})$ in which all of the variables are connected by $\wedge$, satisfies this condition on the automorphism 
if and only if $C_{i}(X_{k})=(x_{1}\wedge x_2\wedge\ldots\wedge x_{\binom{k}{2}})\ or\ (\neg x_{1}\wedge \neg x_{2}\wedge\ldots\wedge \neg x_{\binom{k}{2}})$ 
Among them, a Boolean function which correctly recognizes a $k$-clique is only $(x_{1}\wedge x_{2}\wedge\ldots\wedge x_{\binom{k}{2}})$.\\
Therefore $f_{k}^{restricted}(X_{k})=(x_{1}\wedge x_{2}\wedge\ldots\wedge x_{\binom{k}{2}})$.
\end{proof}

\begin{theorem} For a variable $y (\not\in X_{k})$ representing an edge, $Aut(f_{k}^{restricted}(X_{k},y))\not\cong S_{\binom{k}{2}+1}$.\end{theorem}
\begin{proof}
A method of proof by contradition is used. Suppose if $Aut(f_{k}^{restricted}(X_{k},y))\cong S_{\binom{k}{2}+1}$, then for any input $(X_{k},y)$ 
the return value of $f_{k}^{restricted}(X_{k},y)$ does not change after the permutation on variables. However for two inputs\\
$(X_{k},y) = (x_{1},x_{2},\ldots,x_{\binom{k}{2}},y)=(1,1,\ldots,1,1,0)$ and\\
$(y,x_{2},\ldots,x_{\binom{k}{2}},x_{1})=(0,1,\ldots,1,1,1)(x_{1}\ and\ y\ are\ exchanged)$,\\
the return values of each of these inputs need to be different.
\begin{multline} f_{k}^{restricted}(x_{1},x_{2},\ldots,x_{\binom{k}{2}},y)=f_{k}^{restricted}(1,1,\ldots,1,0)=1\\
\neq f_{k}^{restricted}(y,x_{2},\ldots,x_{\binom{k}{2}},x_{1})=f_{k}^{restricted}(0,1,\ldots,1,1)=0\end{multline}
Therefore $Aut(f_{k}^{restricted}(X_{k},y))\not\cong S_{\binom{k}{2}+1}$.
\end{proof}
Based on these results, the local structure of Boolean function $f_{k}$ can be specified. So next, it is necessary to specify 
its global structure based on its local structure. For $f_{k}(x_{1},x_{2},\ldots,x_{\binom{n}{2}})=C_{1}\vee C_{2}\vee\ldots\vee C_{m}$ 
$(C_{j}(1\le j\le m)$ is not the same as $C_{j}$ used in the discussion above), the following theorem follows.
\begin{theorem}
(A) $C_{j}(1\le j\le m)$ has to contain at least $\binom{k}{2}$ variables.\\
(B) $\binom{k}{2}$ variables in $X_{k}^{i}(1\le i\le\binom{n}{k})$ have to be contained in one of $C_{j}(1\le j\le m)$.\\
(C) After reordering clauses, we can take 
$C_{i}(X_{k}^{i})=f_{k}^{restricted}(X_{k}^{i})=(x_{1}^{i}\wedge x_{2}^{i}\wedge\ldots\wedge x_{\binom{k}{2}}^{i})(1\le i\le\binom{n}{k}).$
\end{theorem}
\begin{proof}(A) Like the discussion in theorem2, a method of proof by contradiction is used. 
If $C_{j}(1\le j\le m)$ contains only $l(<\binom{k}{2})$ variables, then two cases are conceivable.\\
(1) $C_{j}$ is satisfiable if the truth values of $l$ variables are appropriately chosen.\\
(2) $C_{j}$ is not satisfiable for any of the truth values.

In case (1), $C_{j}=1$ for $l$ variables with appropriately chosen truth values. So
\begin{equation}\label{alwaystrue2} f_{k}(x_{1},x_{2},\ldots,x_{\binom{n}{2}})=C_{1}\vee C_{2}\vee\ldots\vee C_{j}\vee\ldots\vee C_{m}=1 \end{equation}
But if the variable not used in $C_{j}$ and the remaining variables take 0, $f_{k}(x_{1},x_{2},\ldots,x_{\binom{n}{2}})$ should return $0$ 
as $X_{k}$ does not form a $k$-clique. This contradicts with (\ref{alwaystrue2}).

In case (2), as $C_{j}=0$\\
$f_{k}(x_{1},x_{2},\ldots,x_{\binom{n}{2}})=C_{1}\vee C_{2}\vee\ldots\vee C_{j-1}\vee C_{j}\vee C_{j+1}\vee\ldots\vee C_{m}$\\
$=C_{1}\vee C_{2}\vee\ldots\vee C_{j-1}\vee 0\vee C_{j+1}\vee\ldots\vee C_{m}$\\
$=C_{1}\vee C_{2}\vee\ldots\vee C_{j-1}\vee C_{j+1}\vee\ldots\vee C_{m}$\\ 
$C_{j}$ does not influence the return value and should be erased. Therefore each one of $C_{j}$ has to contain at least $\binom{k}{2}$ variables.

(B) If $\binom{k}{2}$ variables in $X_{k}^{i}(1\le i\le \binom{n}{k})$ are not contained in any of $C_{j}$, 
then $f_{k}^{restricted}(X_{k}^{i})$ does not contain $\binom{k}{2}$ variables in $X_{k}^{i}$, which contradicts with theorem 5. 
Therefore $\binom{k}{2}$ variables in $X_{k}^{i}(1\le i\le \binom{n}{k})$ have to be contained in one of $C_{j}(1\le j \le m)$.

(C) For clauses which contain $\binom{k}{2}$ variables in $X_{k}^{i}$, suppose if all of them have more than $\binom{k}{2}$ variables.
Then a clause $C_{j}$ which satisfies the above condition takes the forms
\begin{equation}\label{smallest1} C_{j}=(x_{1}^{i}\wedge x_{2}^{i}\wedge\ldots\wedge x_{\binom{k}{2}}^{i}\wedge y_{1}\wedge Y) \end{equation}
(Y is a clause in which variables are connected by $\wedge$).
\begin{equation}\label{smallest2} or\ C_{j}=(x_{1}^{i}\wedge x_{2}^{i}\wedge\ldots\wedge x_{\binom{k}{2}}^{i} \bigwedge_{l=1}(\neg y_{l})) \end{equation}
In (\ref{smallest1}), assigning false values to variables other than \\
$x_{1}^{i},x_{2}^{i},\ldots,x_{\binom{k}{2}}^{i}$ does not produce a clause \\
$f_{k}^{restricted}(X_{k}^{i})=(x_{1}^{i}\wedge x_{2}^{i}\wedge\ldots\wedge x_{\binom{k}{2}}^{i})$ \\
In (\ref{smallest2}), assigning false values to variables other than \\
$x_{1}^{i},x_{2}^{i},\ldots,x_{\binom{k}{2}}^{i}$ produces a clause \\
$f_{k}^{restricted}(X_{k}^{i})=(x_{1}^{i}\wedge x_{2}^{i}\wedge\ldots\wedge x_{\binom{k}{2}}^{i})$ \\
but assigning $(x_{1}^{i}=x_{2}^{i}=\ldots=x_{\binom{k}{2}}^{i}=1,y_{l}=1(1\le l))$ to $C_{j}$ returns $0$, though $C_{j}$ and $f_{k}$ should return $1$.\\
Therefore neither (\ref{smallest1}) nor (\ref{smallest2}) follow. 
So $C_{j}$ has exactly $\binom{k}{2}$ variables in $X_{k}^{i}$, and $C_{j}=(x_{1}^{i}\wedge x_{2}^{i}\wedge\ldots\wedge x_{\binom{k}{2}}^{i})$.
\end{proof}
By theorem7, $f_{k}$ can be expressed in the following form.
\begin{equation} f_{k}(x_{1},x_{2},\ldots,x_{\binom{n}{2}})=(\bigvee_{i=1}^{\binom{n}{k}}f_{k}^{restricted}(X_{k}^{i})) \bigvee_{j=1} C_{j}' \end{equation}
It is necessary to specify the form of $C_{j}'$. 
\begin{theorem} $C_{j}'$ has to contain $\binom{k}{2}$ variables representing edges among $k$ vertices:
$C_{j}'=C_{j}'(X_{k}^{i},\ldots)$ for some $i$ \end{theorem}
\begin{proof}
If $C_{j}'$ returns $0$ for all of the inputs, $C_{j}'$ should be erased in $f_{k}$. 
So it is necessary to consider the case where $C_{j}'$ returns $1$ for some input. 
If no $\binom{k}{2}$ variables in $C_{j}'$ represent edges among $k$ vertices and $C_{j}'$ returns $1$ for some input, 
that contradicts with the fact \\
$\displaystyle f_{k}(x_{1},x_{2},\ldots,x_{\binom{n}{2}})=(\bigvee_{i=1}^{\binom{n}{k}}f_{k}^{restricted}(X_{k}^{i}))\bigvee_{j} C_{j}'$ \\
detects $k$ cliques. Therefore $C_{j}'$ has to contain $\binom{k}{2}$ variables representing edges among $k$ vertices.
\end{proof}
\begin{theorem} $C_{j}'$ has to contain $\binom{k}{2}$ variables representing edges among $k$ vertices ``without NOT gates'':
$C_{j}'=(x_{1}^{i}\wedge x_{2}^{i}\wedge\ldots\wedge x_{\binom{k}{2}}^{i}\wedge \ldots)$ \end{theorem}
\begin{proof}
By theorem8, $C_{j}'=C_{j}'(X_{k}^{i},\ldots)$. As $C_{j}'$ might contain variables representing more than one clique, 
variables in $C_{j}'$ can be written as
\begin{equation} X_{k}^{i_{1}}\cup X_{k}^{i_{2}}\cup\ldots\cup X_{k}^{i_{m1}}\cup y_{1}\cup y_{2}\cup\ldots\cup y_{m2}\ (m1\ge 1) \end{equation}
$y_{1},y_{2},\ldots,y_{m2}$ do not contain variables representing a $k$-clique. Suppose if at least one variable in each $X_{k}^{i}(i=i1,i2,\ldots,i_{m1})$ 
is connected to a NOT gate, then $C_{j}'$ can be written as
\begin{multline}\displaystyle C_{j}'=C_{j}'(X_{k}^{i_{1}}\cup X_{k}^{i_{2}}\cup\ldots\cup X_{k}^{i_{m1}}\cup y_{1}\cup y_{2}\cup\ldots\cup y_{m2}) \\
=(\bigwedge(\neg x_{z1}) \bigwedge x_{z2}\bigwedge (\neg y_{w1})\bigwedge y_{w2})\\
(x_{z1},x_{z2}\in X_{k}^{i}, i=i_{1},i_{2},\ldots,i_{m1}, 1\le w1,w2\le m2) \end{multline}
As each $X_{k}^{i}(i=i1,i2,\ldots,i_{m1})$ has to contain at least one $x_{z1}$, an input $x_{z1}=0$, $x_{z2}=1$, $y_{w1}=0$, $y_{w2}=1$ 
does not contain a $k$-clique though 
\begin{equation} C_{j}'(x_{z1}=0,x_{z2}=1,y_{w1}=0,y_{w2}=1)=1 \end{equation} 
Therefore $C_{j}'$ has to contain $\binom{k}{2}$ variables in $X_{k}^{i}$ without NOT gates. 
So $C_{j}'=(x_{1}^{i}\wedge x_{2}^{i}\wedge\ldots\wedge x_{\binom{k}{2}}^{i}\wedge \ldots)$.
\end{proof}
By theorem9, $f_{k}$ can be specified as follows.
\begin{equation}\label{fk1} f_{k}(x_{1},x_{2},\ldots,x_{\binom{n}{2}})=(\bigvee_{i=1}^{\binom{n}{k}}f_{k}^{restricted}(X_{k}^{i}))\bigvee_{j} C_{j}' \end{equation}
\begin{equation} C_{j}'=(x_{1}^{i}\wedge x_{2}^{i}\wedge\ldots\wedge x_{\binom{k}{2}}^{i}\wedge Y)\ (Y\ is\ a\ clause)\end{equation}
As indicated above, conversion into the disjunctive normal form sometime results in an explosion in the formula.
So it is necessary to determine the minimum size of $f_{k}$ among its logical equavalences.
\begin{theorem}For $f_{k}(x_{1},x_{2},\ldots,x_{\binom{n}{2}})$ for the $k$-clique problem, the circuit size of $f_{k}$ is larger than $(\frac{n}{k})^{k}-1$.\end{theorem}
\begin{proof}
Using the absorption law,
\begin{equation} A\vee(A\wedge B)=A\wedge(1\vee B)=A \end{equation}
(\ref{fk1}) can be expressed in a smaller circuit by compressing
\begin{multline} f_{k}^{restricted}(X_{k}^{i})\vee C_{j}'\\
=f_{k}^{restricted}(X_{k}^{i})\vee (f_{k}^{restricted}(X_{k}^{i})\wedge Y)=f_{k}^{restricted}(X_{k}^{i}) \end{multline}
\begin{equation}\label{fk2} f_{k}(x_{1},x_{2},\ldots,x_{\binom{n}{2}})=\bigvee_{i=1}^{\binom{n}{k}}(f_{k}^{restricted}(X_{k}^{i})) \end{equation}
(Of course, the compressed form (\ref{fk2}) satisfies theorem1 and theorem6.)\\
For $i,j(i\neq j,1\le i,j\le \binom{k}{2})$, edge sets $X_{k}^{i}$ and $X_{k}^{j}$ have at most $l(<\binom{k}{2})$ elements in common.
The only way to express (\ref{fk2}) in a smaller circuit is to apply the distributive law
\begin{equation}\label{distributive} (A\wedge B)\vee(A\wedge C)=A\wedge(B\vee C) \end{equation}
As the number of OR gates in (\ref{fk2}) cannot be reduced by converting it into its logical equivalences using (\ref{distributive}), 
the size of the Boolean circuit is larger than the number of OR gates expressed as a Boolean circuit of (\ref{fk2}), not as a disjunctive normal form. 
Therefore the size of $f_{k}$ as a Boolean circuit is larger than $\binom{n}{k}-1>(\frac{n}{k})^{k}-1$.
\end{proof}
By theorem10, for $k$ in $3<k<{n}^{\frac{1}{4}}$, this proves $P/poly \neq NP$ and $P \neq NP$.

\subsection{Proof that this strategy is non-Naturalizable}
In the paper \cite{Natural}, the proof is natural or natularizable if it satisfies the following three conditions, constructivity, largeness, and usefulness.\\
{\em''
Formally, by a combinatorial property of Boolean functions we will mean a set of Boolean functions $\{C_{n}\subseteq F_{n}| n\in\omega \}$. 
Thus, a Boolean function $f_{n}$ will possess property $C_{n}$ if and only if $f_{n}\in C_{n}$. 
(Alternatively, we will sometimes find it convenient to use function notation: $C_{n}(f_{n})=1$ if $f_{n}\in C_{n}$; $C_{n}(f_{n})=0$ if $f_{n}\not\in C_{n}$.) 
The combinatorial property $C_{n}$ is natural if it contains a subset $C_{n}^{*}$ with the following two conditions:\\

Constructivity.   The predicate $f_{n}\overset{?}{\in}C_{n}^{*}$ is computable in P. Thus $C_{n}^{*}$ is computable in time which is polynomial 
in the truth table of $f_{n}$;\\

Largeness.     $|C_{n}^{*}| \ge 2^{-O(n)}|F_{n}|$\\
A combinatorial property $C_{n}$ is useful against $P/poly$ if it satisfies:\\

Usefulness.    The circuit size of any sequence of functions $f_{1},f_{2},\ldots,f_{n},\ldots$, where $f_{n}\in C_{n}$, is super-polynomial; i.e., 
for any constant $k$, for sufficiently large $n$, the circuit size of $f_{n}$ is greater than $n^{k}$.
``}\\
The proof strategy used in this paper is to specify the Boolean function $f_{k}$ as (\ref{fk2}). Of course, a Boolean function
\begin{equation} (\ref{fk2})\vee(x_{1}^{1}\wedge\neg x_{1}^{1})=\bigvee_{i=1}^{\binom{n}{k}}f_{k}^{restricted}(X_{k}^{i})\vee(x_{1}^{1}\wedge\neg x_{1}^{1})\end{equation}
can also recognize $k$-cliques correctly, but the essential part of the function is (\ref{fk2}) bacause not a clause in (\ref{fk2}) cannot be erased.
As the aim of this paper is to determine the circuit lower bound, not only logical equivalence but also impossibility to erase a clause should be 
the combinatorial property(If erasable clauses are added, the size gets larger than its strict lower bound). 
So it is reasonable to define the combinatorial property $C_{n}$ as \\
``$f\in C_{n}$ if and only if $f$ is logically equivalent to $f_{k}$ (\ref{fk2}) and not a clause in $f$ cannot be erased 
when expressed in the disjunctive normal form.''\\
Boolean functions satisfying this property can be given by adding double NOT gates such as\\
$\displaystyle f_{k}(x_{1},x_{2},\ldots,x_{\binom{n}{2}})=\bigvee_{i=1}^{\binom{n}{k}}f_{k}^{restricted}(X_{k}^{i})$\\
$\displaystyle =(x_{1}^{1}\wedge x_{2}^{1}\wedge\ldots\wedge x_{\binom{k}{2}}^{1})\bigvee_{i=2}^{\binom{n}{k}}f_{k}^{restricted}(X_{k}^{i})$\\
$\displaystyle =\neg((\neg x_{1}^{1})\vee (\neg x_{2}^{1})\vee\ldots\vee (\neg x_{\binom{k}{2}}^{1}))\bigvee_{i=2}^{\binom{n}{k}}f_{k}^{restricted}(X_{k}^{i})$\\
Therefore the total number of Boolean function satisfying this property is at most $2^{\binom{n}{k}}$.
As the total number of Boolean functions with $n$ variables is $2^{2^{n}}$, 
\begin{equation} \frac{|C_{n}^{*}|}{|F_{n}|}=\frac{2^{\binom{n}{k}}}{2^{2^{n}}} << 2^{-O(n)} \end{equation}
This violates the largeness condition in Natural proof. Therefore this strategy does not conflict with the widely believed conjecture 
on the existence of pseudorandom functions.

\section{Relationship to other open problems in Computational Complexity Theory}
As $P\neq NP$ and $NP\subseteq PH \subseteq PSPACE$, $P\neq PH$ and $P\neq PSPACE$.
Among problems in $NP$, complexity classes of the integer factorization problem, the discrete logarithm problem and the graph isomorphism problem
\cite{luks}\cite{ladner:intermediate}\cite{gipdisease}\cite{uwegiplow} remain open for many years. 
It is already known that \\
(1)if the decision version of the integer factorization problem is in NP-complete, 
then NP=co-NP and the polynomial hierarchy will collapse to its first level.\\
(2)if the graph isomorphism problem is in NP-complete, then the polynomial hierarchy will collapse to its second level.\\
As a collapse of polynomial hierarchy seems unlikely to happen under $P \neq NP$, they seem unlikely to be in NP-complete.
Furthermore the circuit lower bounds of the integer factorization problem, the discrete logarithm problem and the graph isomorphism problem 
cannot be obtained by the proof strategy used in this paper, because \\
(1)the integer factorization problem and the discrete logarithm problem have neither global symmetry nor local symmetry to specify their structure.\\
(2)the global symmetry and local symmetry of graph isomorphism problem are hard to determine in general.\\

Whether or not NP-complete problems can be solved by quantum computers in polynomial time remains open.

\section{Conclusions}
By interpreting a Boolean circuit as a graph, the global symmetry and the local symmetry of variables in the circuit is discussed in this paper. 
The small number of global symmetry and the large number of local symmetry in the circuit which computes $f_{k}$ can establish 
the exponential circuit lower bound for a NP-complete problem, which means $P/poly\neq NP$ and $P\neq NP$. 

Even if the same strategy is used, the computational complexity classes of the integer factorization problem, the discrete logarithm problem and 
the graph isomorphism problem remain open. Furthermore whether or not NP-complete problems can be solved by quantum computers in polynomial time remain open.

As NP-complete problems turn out to be impossible to solve in polynomial time by a classical computer, heuristic approaches or 
algorithms for restricted types of inputs need to be developed for NP-complete problems.

\bibliographystyle{abbrv}
\bibliography{circuitAndSymmetry.bib}

\end{document}